\documentclass[conference]{IEEEtran}

\usepackage{graphicx}
\usepackage{wrapfig}
\usepackage{epsfig}
\usepackage{graphicx}
\usepackage{color}
\usepackage{mdframed}
\usepackage{mathrsfs} 
\usepackage{amsfonts}
\usepackage{latexsym}
\usepackage{amsmath}
\usepackage{amssymb}
\usepackage{color}
\usepackage{float} 
\usepackage{floatflt} 
\usepackage{caption}
\usepackage{subcaption}

\usepackage{pdfpages}

\newcommand{\suppress}[1]{}

\newtheorem{theorem}{Theorem}[section]

\newtheorem{cor}{Corollary}[section]

\newtheorem{claim}{Claim}[section]

\newtheorem{definition}{Definition}[section]

\newcommand{\hM}{{\hat{M}}}

\newcommand{\dmax}{d_{\rm max}}
\newcommand{\davg}{d_{\rm avg}}
\newcommand{\Cmax}{{\cal C}_{\rm max}}
\newcommand{\Cavg}{{\cal C}_{\rm avg}}

\def\e{\varepsilon}

\def\cX{\mbox{$\cal{X}$}}
\def\cY{\mbox{$\cal{Y}$}}

\def\cC{\mbox{$\cal{C}$}}

\def\cM{\mbox{$\cal{M}$}}

\def\cN{\mbox{$\cal{N}$}}

\def\e{\varepsilon}

\def\bx{{\bf x}}

\def\by{{\bf y}}

\newcommand{\bX}{{{\bf X}}}

\def\01{\{0,1\}}

\newcommand{\remove}[1]{}

\begin{document}
\IEEEoverridecommandlockouts
\title{On the Capacity Advantage of a Single Bit}

\author{Michael Langberg\ \ \ \ \ \ \ \ \  \ \ \ \ \ \ \ \ \ 
Michelle Effros
\thanks{M. Langberg is with the Department of Electrical Engineering at The State University of New York at Buffalo.  
Email : \texttt{mikel@buffalo.edu}}
\thanks{M. Effros is with the Department of Electrical Engineering at the California Institute of Technology.
Email : \texttt{effros@caltech.edu}}
\thanks{This material is based upon work supported by NSF grants CCF-1018741 and CCF-1016671.}
}

\maketitle

\begin{abstract}
In this work we study the capacity advantage 
achieved by adding a {\em single bit} of communication -- 
not a link of capacity 1 but a single bit over all time --
to a memoryless network.
Specifically, we present a memoryless network 
in which adding a single bit of communication 
strictly increases the capacity region.
\end{abstract}

\section{Introduction}
The {\em edge removal problem}, defined and studied in \cite{HEJ:10,JEH:11}, 
aims to quantify the loss in capacity 
that results from the removal of a single edge 
(i.e., a point-to-point channel) 
from a given (possible noisy) communication network. 
At first glance, it may seem that the removal of an edge 
with edge capacity $\delta>0$ 
should change the network capacity by an amount 
that tends to zero as $\delta$ tends to zero.  
The following examples 
demonstrate that this is not always the case.

In~\cite{guruswami2003list, langberg2004private}, 
Guruswami and Langberg present simple examples of 
point-to-point channels with memory for which 
the capacity of a network $\cal N$ 
containing the channel 
and a {\em side-information} edge of capacity $\delta$ 
from the transmitter to the receiver 
is described by a function $C_{\cal N}(\delta)$ 
that is not continuous at $\delta=0$. 

While the discontinuity of $C_{\cal N}(\delta)$ 
in~\cite{guruswami2003list, langberg2004private} 
can be attributed to the memory in the point-to-point channel,~\cite{NEL:16} 
shows a similar phenomenon for a memoryless network.
Specifically,~\cite{NEL:16} 
studies communication over a network $\cal N$ 
containing a memoryless Multiple Access Channels (MAC) 
and a $\delta$-capacity {\em cooperation} edge 
carrying information from a ``cooperation facilitator'' to the two transmitters.  
Once again, the capacity $C_{\cal N}(\delta)$ exhibits a discontinuity at $\delta=0$. 
More specifically, \cite{NEL:16} shows that 
for any MAC whose average- and maximal-error capacities differ, 
adding a cooperation facilitator with output capacity $\delta$ 
results in a network $\cal N$ 
whose maximal-error capacity $C_{\cal N}(\delta)$ 
exhibits the described discontinuity.
Dueck's 2-way {\em contraction channel}~\cite{Dueck:78} 
is an example of a MAC with the described property.

In this work, 
we take the edge removal problem to an extreme, 
seeking to understand the effect of removing 
not a $\delta$-capacity edge for arbitrarily small $\delta>0$ 
but rather an edge that can only carry 1 bit of communication (over all time). 
We wish to understand whether there exist networks 
for which the removal of a single bit of communication 
can strictly change the capacity region.  Our study focuses on noisy networks. 
We discuss noiseless networks, i.e., network coding, at the end of this section.

For noisy channels with memory, the answer is immediate. 
It is not difficult to construct a two-state point-to-point channel with memory 
for which a single bit of feedback changes the network capacity.
For example, 
consider a binary symmetric channel 
whose error probability $\theta$ is chosen at random 
and then fixed for all time.  
If $\theta$ equals $\theta_i$, $i\in\{1,2\}$, 
with positive probability $p_i$ and $0\leq\theta_1<\theta_2\leq1/2$
then a single bit of feedback from the receiver to the transmitter 
asymptotically suffices to increase the capacity.

For memoryless channels, the question is far more subtle.
Its solution is the subject of this work. 
In Section~\ref{sec:model}, 
we formalize the notion of a single bit of communication 
by defining the ``1-bit channel" 
and the corresponding notion of capacity to accommodate networks 
with 1-bit channels.
In Section~\ref{sec:main}, 
we demonstrate the existence of a network 
for which the removal of a single 1-bit channel 
changes the capacity region significantly. 
We employ the contraction channel of Dueck \cite{Dueck:78} in our construction.

We note that the edge removal problem has seen significant studies in the noiseless setting of network coding as well.
In contrast to the noisy communication setting, in network coding instances, the question whether there exists networks for which the removal of edges of negligible capacity (in the block length $n$) causes a strictly positive loss in rate is an open problem. For several network coding instances studied in \cite{HEJ:10,JEH:11} it is shown that the removal of an edge of capacity $\delta$ may decrease the rate of communication between each source-receiver pair by at most $\delta$. These instances include networks with collocated sources, networks in which we are restricted to perform linear encoding, networks in which the edges removed are connected to receivers with no out going edges, and additional families of network coding instances. 
However, whether there exist network coding instances with a $\delta$ capacity edge for which the capacity region $C(\delta)$ is not continuous at $\delta=0$ remains an intriguing open problem connected to a spectrum of (at times seemingly unrelated) questions in the context of network communication, e.g., \cite{CG10,chan2014network, LE11, LE:12, WLE13, WLE:15,WLE:16}. One can study the effect of removing 1-bit channels in the noiseless setting of network coding. It is tempting to believe that the removal of a 1-bit channel cannot effect the capacity region of network coding instances, however, the problem is left open in this work.

%

\section{Model}
\label{sec:model}

We wish to consider memoryless networks 
enhanced by the addition of one or more 1-bit channels. 
We begin by defining memoryless and 1-bit channels 
and then define codes and capacities 
appropriate for networks that combine them.

The following definitions are useful 
in the description that follows.
For any positive constant $c$, 
let $[c]=\{1,\ldots,\lfloor c\rfloor\}$.
For any sets $A\subseteq B$ and 
any values $x_i$ or sets $\cX_i$ indexed by $i\in B$, 
let $x_A=(x_i:i\in A)$ and $\cX_A=\prod_{i\in A}\cX_i$.

\subsection{Networks}

An $m$-node {\em memoryless network} 
is described by a triple 
\[
\left(\cX_{[m]},p(y_{[m]}|x_{[m]}),\cY_{[m]}\right),
\]
where $x_i\in\cX_i$ and $y_i\in\cY_i$ 
represent the network's inputs from and outputs to node $i$.
Thus, at each time $t$ node $i$ 
transmits a network input $X_{i,t}$ and 
receives a network output $Y_{i,t}$ 
governed by the statistical relationship $p(y_{[m]}|x_{[m]})$ 
relating all network inputs to all network outputs at time $t$.

A {\em 1-bit channel} is a point-to-point channel 
from some node $i$ to another node $j$ in the network.  
Unlike the memoryless network above,
which can carry the same amount of information 
in every time step, 
the 1-bit channel can carry only one bit of information in total.  
That information is transmitted 
at some fixed time $\tau$ chosen in the code design.  
Given $\tau$, the 1-bit channel effectively acts like 
a collection of independent, single-time-step memoryless channels, 
where the channel at time $t$ is 
\[
\begin{array}{cl}
\left(\{0,1\},1(y_{j,t}=x_{i,t}),\{0,1\}\right) & \mbox{if $t=\tau$} \\
\left(\{0\},1(y_{j,t}=x_{i,t}),\{0\}\right) & \mbox{if $t\neq\tau$}.
\end{array}
\]
Here we use an alphabet $\{0\}$ of size 1 to represent the fact that the 1-bit channel is inactive when $t\neq\tau$.
\subsection{Codes}

Consider a network $\cN$
that combines a memoryless channel with one or more 1-bit channels. 
To capture the 1-bit channels, 
we hence forward allow alphabets $\cX_{i,t}$ and $\cY_{i,t}$ 
to vary with $t$ since node $i$ may transmit information to 
or receive information from one or more 1-bit channels.

A {\em blocklength-$n$ code} $(\tau,X,\hM)$ 
operates the network $\cN$ over $n$ time steps 
with the goal of transmitting a message $M_i\in\cM_i$ 
from each node $i\in[m]$ to all receivers $j$ in demand set $D_i\subseteq[m]\setminus\{i\}$.
By assumption, $\cM_i=[2^{nR_i}]$, 
$R_i$ equals zero if and only if $D_i=\emptyset$, 
and $M_{[m]}$ is distributed uniformly on $\cM_{[m]}$.  
Code $(\tau,X,\hM)$ is defined by 
its activation vector 
\[
\tau=(\tau_e)_{e\in\alpha},
\]
encoding functions 
\[
X = (X_{i,t})_{i\in[m],t\in[n]}, 
\]
and decoding functions 
\[
\hM=(\hat{M}_{j,i})_{j\in[m],i\in D_j}.
\]
Here, $\alpha$ represents the set of 1-bit channels in $\cN$.
For each 1-bit channel $e$, 
constant $\tau_e\in[n]$ 
describes the activation time of $e$.
For each $i$ and $t$, 
encoder 
\[
X_{i,t}:\cY_{i,[t-1]}\times \cM_i\rightarrow \cX_{i,t}
\]
maps node $i$'s previously received network outputs and 
outgoing message 
to the time-$t$, node-$i$ network input $X_{i,t}$. 
For each message $M_j$ and node $i$ that demands $M_j$, 
the decoder 
\[
\hM_{j,i}:\cY_{i,[n]}\times\cM_i\rightarrow\cM_j
\]
maps node $i$'s  received network outputs 
and outgoing message 
to a reconstruction of message $M_j$.  

\subsection{Capacities}

For the purpose of defining capacity, 
it is useful to characterize a code $(\tau,X,\hM)$ 
by its blocklength $n$, rate $R_{[m]}$, 
and error probability.  
The literature contains several definitions of error probability, 
most notably the {\em maximal error probability}
\[
\dmax(\tau,X,\hM) = \max_{j\in[m],i\in D(j)}\Pr(\hM_{j,i}\neq M_j) 
\]
and the {\em average error probability} $\davg(\tau,X,\hM)$:
\[
\frac1{|\cM_{[m]}|}\sum_{j\in [m],w_j\in \cM_j} \Pr(\exists i\in D_j: \hM_{j,i}\neq w_j\mid M_j=w_j).
\]
A blocklength-$n$, rate-$R_{[m]}$ code $(\tau,X,\hM)$ 
is called a max-$(\e,R_{[m]},n)$ solution if 
$\dmax(\tau,X,\hM)\leq \e$ 
and an avg-$(\bar{\e},R_{[m]},n)$ solution if 
$\davg(\tau,X,\hM)\leq \bar{\e}$.

\begin{definition}[Maximal-Error Capacity]
The {maximal-error capacity region} $\cC_{\max}(\cN)$ of 
network $\cN$ is the closure of all rate vectors $R_{[m]}$ 
such that for any $\e>0$, there exists a max-$(\e,R_{[m]},n)$ solution 
for all $n$ sufficiently large.
\end{definition}

\begin{definition}[Average-Error Capacity]
The {average-error capacity region} $\cC_{\rm avg}(\cN)$ 
of network $\cN$ is the closure of all rate vectors $R_{[m]}$ 
such that for any $\bar{\e}>0$, 
there exists an avg-$(\bar{\e},R_{[m]},n)$ solution 
for all $n$ sufficiently large.
\end{definition}

\subsection{Remarks}

Some remarks are in order. 
Notice that for $R^{[m]}$ 
to be included in (the interior of) a capacity region 
$\Cmax(\cN)$ (similarly for $\Cavg(\cN)$), 
we require the existence of max-$(\e,R^{[n]},n)$ codes $(\tau,X,\hM)$ 
for {\em all} sufficiently large blocklengths. 
Without such a requirement, 
the addition of 1-bit channels can have a significant effect 
on the achievable rate for small blocklengths $n$. 
Take for example an empty network $\cN$ 
enhanced with a single 1-bit channel $e$ to obtain $\cN^+$. 
For $n=1$, there is a clear difference 
between the rates achievable on $\cN$ and $\cN^+$. 
We therefore require the blocklength $n$ to grow without bound 
to capture the idea that the channel $e$ can carry only one bit 
{\em over all time} rather than that $e$ can carry one bit in a small time window $[n]$. 
Notice further that our definitions 
require each 1-bit channel $e\in\alpha$
to be active at a constant time $\tau_e$ 
independent of the messages $M_{[m]}$.  
This is important since it prevents codes 
that might use timing to convey information 
about the messages. 

\section{Main result}
\label{sec:main}

The main question we ask in this work is whether there exists a network $\cN$ and a 1-bit channel $e$ such that adding $e$ to $\cN$ yields a new network $\cN^*$ with a strictly larger capacity region. Namely, with
$$
\cC(\cN^{*}) \ne \cC(\cN).
$$
We answer the question for maximal error below.
We leave the question in the context of average error open in this work.\footnote{The results of \cite{NEL:16} that exhibit capacity $C_{\cal N}(\delta)$ with a discontinuity at $\delta=0$ hold for maximal error only. The average error case remains open.}

To prove our result, we start by constructing a network $\cN^+$ 
by combining Dueck's memoryless MAC~\cite{Dueck:78} 
with a pair of memoryless point-to-point channels 
and a collection $\alpha$ of 1-bit channels.
The construction, described below, is depicted in Figure~\ref{fig:plus}.

Dueck's MAC has a pair of transmitters, nodes~1 and~2, 
and a single receiver, node~3.  
The input and output alphabets are 
\[
\begin{array}{rclrcl}
\cX_1 & = & \{a,b,A,B\} & \cY_1 & = & \{0\} \\
\cX_2 & = & \{0,1\}        & \cY_2 & = & \{0\} \\
\cX_3 & = & \{0\}  & \cY_3 & = & \{a,b,c,A,B,C\}\times\{0,1\} 
\end{array}
\]
Since nodes~1 and~2 receive no outputs from the channel 
and node~3 has no input to the channel (here denoted by alphabets of size 1)
we simplify the notation from Section~\ref{sec:model} 
to describe the channel as $\cN_D=(\cX_1\times\cX_2,p(y_3|x_1,x_2),\cY_3)$.  
Where $p(y_3|x_1,x_2)=1(y_3=W(x_1,x_2))$ 
captures the channel's deterministic behavior with output
\[
W(x_1,x_2) = \left\{
\begin{array}{ll}
(c,0) & \mbox{if $(x_1,x_2)\in\{(a,0),(b,0)\}$} \\
(C,1) & \mbox{if $(x_1,x_2)\in\{(A,1),(B,1)\}$} \\
(x_1,x_2) & \mbox{otherwise.}
\end{array}
\right.
\]

We next add to Dueck's MAC $\cN_D$ 
a fourth node, here called the ``cooperation facilitator'' (CF), 
and a pair of memoryless, point-to-point channels.  
We denote the resulting network by $\cN_0$.  
The first memoryless channel, from node~1 to node~4, 
is a noiseless channel of capacity 2.  
The second, from node~2 to node~4, is a noiseless channel of capacity 1.
Given a pair of messages originating at nodes~1 and~2 
and a single receiver at node~3 ($D_1=D_2=\{3\}$ and $D_3=D_4=\emptyset$), 
the capacity region of this modified channel and Dueck's MAC are identical
($\Cmax(\cN_0)=\Cmax(\cN_D)$).

Finally, we build network $\cN^+$ from memoryless network $\cN_0$
by adding $k$ 1-bit channels, $\{f_1,\dots,f_k\}$, from the CF (node~4) to receiver node~3, and by adding a single 1-bit channel, $e$, 
from the CF to node~1.  

\begin{figure}[h!]
\centering
\includegraphics[scale=0.35]{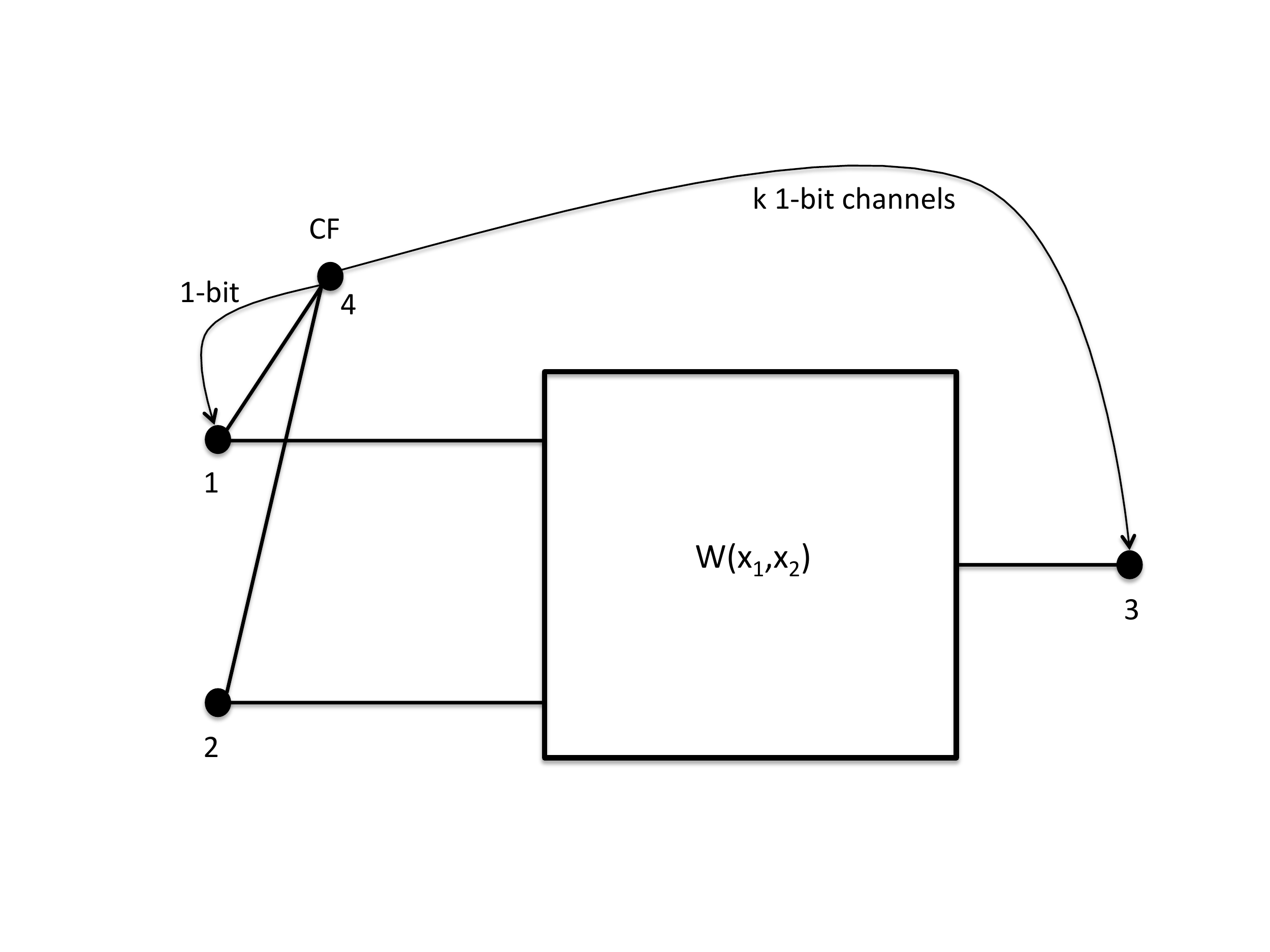}
\caption{Network $\cN^+$. 
Transmitter nodes~1 and~2 wish to send a pair of independent messages 
to the receiver node~3.  The cooperation facilitator, node~4, is labeled by $CF$.  
Node~1 receives a 1-bit input from the $CF$. 
Receiver node~3 receives an additional $k$-bit input from the $CF$.
}
\label{fig:plus}
\end{figure}



Our main result shows that for some finite $k$, 
$$
\Cmax(\cN^+) \ne \Cmax(\cN_0).
$$
That is, removing the 1-bit edges leaving node CF changes the capacity region of the network.
\begin{theorem}
\label{the:main}
For some finite constant $k$, 
$$
(1.24,1) \in \Cmax(\cN^+)
$$
while,
$$
(1.19,0.97) \not\in \Cmax(\cN_0)
$$
\end{theorem}

As a corollary of Theorem~\ref{the:main} we have:
\begin{cor}
\label{cor:main}
There exists a network $\cN$ and a single 1-bit channel $e$ such that the network $\cN^*$ obtained by adding $e$ to $\cN$ satisfies:
$$
\Cmax(\cN^{*}) \ne \Cmax(\cN).
$$
\end{cor}

\begin{proof} (of Corollary~\ref{cor:main})
Consider adding the $(k+1)$ 1-bit channels of $\cN^{+}$ one by one to the network $\cN_0$. 
Specifically, let $\cN_1$ be the network $\cN_0$ enhanced by one of the 1-bit channels of $\cN^{+}$ (which channel precisely is of no significance to the current proof; it may be chosen arbitrarily).
Similarly, for $i=2,\dots,k+1$,  let $\cN_i$ be the network $\cN_{i-1}$ enhanced by one of the 1-bit channels of $\cN^{+}$ that do not yet appear in  $\cN_{i-1}$. 
It holds that
$$
\Cmax(\cN_0) \subseteq \Cmax(\cN_1) \subseteq \Cmax(\cN_2) \dots \subseteq \Cmax(\cN_{k+1}).
$$
Since $\Cmax(\cN_0) \ne \Cmax(\cN^+)$ by Theorem~\ref{the:main}, there must be an $i$ for which $\Cmax(\cN_i) \ne \Cmax(\cN_{i+1})$. The corollary then follows by setting $\cN$ equal to $\cN_i$  and $e$ to be the 1-bit channel added to $\cN_i$ to obtain $\cN_{i+1}$.
\end{proof}

\subsection{Proof of Theorem~\ref{the:main}}

First, we note that it is proven in \cite{Dueck:78} that for $p \in [0,1/2]$ 
$$
(H(1/3)+2/3-p,H(p)) \not\in \Cmax(\cN_D)=\Cmax(\cN_0).
$$
Setting $p=0.4$, which gives $H(p)=0.97$, we note that $(1.19,0.97) \not\in \Cmax(\cN_0)$.

Let $\delta>0$ be a sufficiently small constant and let $k=\lceil \log{(8/\delta)} \rceil+1$ (so that $\cN^+$ adds $\lceil \log{(8/\delta)} \rceil+2$ 1-bit channels to $\cN_0$).
We next prove that
$$
(1.25-0.5\delta,1) \in \Cmax(\cN^{+})
$$
by demonstrating the existence,
for any sufficiently large $n$, of a zero-error $(2n+1)$-block length encoding scheme for network $\cN^{+}$ of rate asymptotically close to $(1.25-0.5\delta,1)$. 
We employ a coding scheme that describes the message sets $\cM_1$ and $\cM_2$ in two parts, here denoted by
$\cM_1 = \cM_{1}^1 \times \cM_1^2=[2^n]\times [2^{\frac{3n}{2}-\delta n}]$, and $\cM_2 = \cM_2^1 \times\cM_2^2=[2^{n}]\times[2^{n}]$. 
The first part of each message is described in the first $n$ time steps. All 1-bit channels are activated at time $n+1$. The second part of each message is then described at time steps $n+2$ through $2n+1$.

A rough description of our scheme follows.
In the first phase messages  $M_1^1 \in \cM_1^1$ and $M_2^1 \in \cM_2^1$ are communicated through the network at rate $(1,1)$. Transmitter node~1 and receiver node~3 cannot use in their encoding/decoding functions information from their incoming 1-bit channels since the 1-bit channels have not been activated yet.
During this first phase, the codewords corresponding to messages $M_1^2$ and $M_2^2$ are forwarded on the edges $(1,CF)$ and $(2,CF)$ to the cooperation facilitator $CF$. 
At time step $n+1$ the $CF$, knowing $M_1^2$ and $M_2^2$, computes and sends the outputs to the 1-bit channels.
Finally, during time steps $n+2$ to $2n+1$, messages $M_1^2$ and $M_2^2$ are communicated through the network. Our scheme in the second phase is described in detail below. The strategy  strongly relies on the information received at transmitter node~1 and receiver node~3 from the 1-bit channels leaving the $CF$. The communication rate in the second phase is $(3/2 - \delta,1)$. All in all, the total rate over both phases tends to $(1.25-0.5\delta,1)$ (as $n$ tends to infinity).

\subsubsection{Phase 1 (time steps 1 to $n$)} In phase 1, the codebook for $\cM_1^1$ equals the set of all vectors in $\{a,A\}^n$. The codebook for $\cM_2^1$ equals the set of all the vectors in $\{0,1\}^n$. 
It is not hard to verify that for any $x_1 \in \{a,A\}$ and $x_2 \in \{0,1\}$ one can decode (with zero error) the value of $x_1$ and $x_2$ from $W(x_1,x_2)$. This implies a rate of $(1,1)$ in the first $n$ time steps.
In addition to communicating information regarding $M_1^1$ and $M_2^1$, in this first phase, the $n$ time steps are used to forward the codewords for messages $M_1^2$ and $M_2^2$ (to be specified below) from transmitter nodes 1 and 2 to the $CF$. 

\subsubsection{1-bit channel activation phase (time step $n+1$)} The 1-bit channels in $\cN^+$ are all activated at time step $n+1$. 
Up to time $n+1$, the $CF$ receives the codewords for messages $M_1^2$ and $M_2^2$ from transmitter nodes 1 and 2, and may compute the bits to be transmitted on its outgoing 1-bit  channels. The codewords for messages $M_1^2$ and $M_2^2$ and the functions to be computed for each outgoing 1-bit channel are specified in the description of phase 2.

\subsubsection{Phase 2 (time steps $n+2$ to $2n$)}
In phase 2, the codebook for $\cM_2^2$ (as before) contains all vectors in $\{0,1\}^n$.
In this phase, the codebook for transmitter node~1 contains $2|\cM_1^2|$ codewords.
We here multiply the size of $\cM_1^2$ by 2 as transmitter node~1 may use in its encoding process not only the input message $M_1^2$ but also the input bit $b$ received at time step $n+1$ from the $CF$ over the 1-bit channel.
The first $|\cM_1^2|$ codewords for transmitter node~1
are a set of distinct vectors from $\cX_1^n$, here denoted by $\{\bx_1(w)\}$ where $w=1, \dots, |\cM_1^2|$.
The set $\{\bx_1(w)\}_{w}$ is chosen uniformly at random from the collection of all subsets of $\cX_1^n$ of size $|\cM_1^2|$, i.e., from $\{A \subset \cX_1^n \mid |A|=|\cM_1^2|\}$.
The second $|\cM_1^2|$ codewords for transmitter node~1 are the codewords $\bar{\bx}_1(w)$ for $w=1,\dots,|\cM_1^2|$ where $\bar{\bx}_1(w)$ is a {\em toggled} version of $\bx_1(w)$ in which each lowercase entry of $\bx_1(w)$ is converted to an uppercase value (of the same letter) and vica versa. 
Given message $M_1^2$, node~1 transmits $\bx_1(M_1^2)$ if it receives $b=0$ on its 1-bit channel and transmits  
$\bar{\bx}_1(M_1^2)$ if it receives $b=1$ on its 1-bit channel.

As specified previously, during phase 1 of our communication, the codewords corresponding to $M_1^2$ and $M_2^2$ are transmitted to the $CF$. 
Specifically, node~1 transmits $\bx_1(M_1^2)$ to the $CF$ and node~2 transmits $\bx_2(M_2^2)$ to the $CF$.
Before defining the encoding function of the $CF$ and the decoding function of receiver node~3 we set some notation.

Let $\cY_3=\cY_{3,1} \times \cY_{3,2} = \{A,B,C,a,b,c\}\times \{0,1\}$ be the output alphabet at receiver node~3.
Let $\by_{3,1}=(y_{3,1;1},\dots,y_{3,1;n}) \in {\cY^n_{3,1}}$ and $\by_{3,2}=(y_{3,2;1},\dots,y_{3,2;n}) \in \cY^n_{3,2}$
be the $n$ symbols received in phase 2 at node~3.
Specifically, given transmitted codewords $\bx_1$ and $\bx_2$, we have $(\by_{3,1},\by_{3,2})=W^n(\bx_1,\bx_2)$, where $W^n$ is the blocklength-$n$ extension of $W$. We use notation $\by_{3,1}=W^n_1(\bx_1,\bx_2)$ and $\by_{3,2}=W^n_2(\bx_1,\bx_2)$, where $W(x_1,x_2)=(W_1(x_1,x_2),W_2(x_1,x_2))$.

Note that $W_2(\bx_1,\bx_2)=\bx_2$ for any $\bx_1$ and $\bx_2$, thus $W^n(\bx_1,\bx_2)=(\by_{3,1},\bx_2)$.
Let $C(\by_{3,1})=|\{i \mid \by_{3,1;i} \in \{c,C\}\}|$ be the number of symbols in $\by_{3,1}$ that equal $c$ or $C$. 
If $\bX_1(\by_{3,1},\bx_2)$ is the set of elements $\bx_1 \in \cX_1^n$ for which $W^n(\bx_1,\bx_2)=(\by_{3,1},\bx_2)$
then $|\bX_1(\by_{3,1},\bx_2)|=2^{C(\by_{3,1})}$. 
As one can decode $\bx_2$ given $\by_{3,1}$ (this follows directly from the definition of $W$), we henceforth use the notation 
$\bX_1(\by_{3,1})$ instead of 
$\bX_1(\by_{3,1},\bx_2)$ respectively.

We say that $\by_{3,1}$ is {\em good} if $|\bX_1(\by_{3,1})| \leq 2^{n/2}$. 
In our encoding scheme for phase 2, we would like to guarantee that $\by_{3,1}$ is always good.
This is accomplished using the information $b$ from the 1-bit channel sent between the $CF$ and node~1.
Namely, let $\bx_1$ and $\bx_2$ be any pair of codewords transmitted to the $CF$ during phase 1 of our communication. 
Given $\bx_1$ and $\bx_2$, the $CF$ may compute $W^n(\bx_1,\bx_2)=(\by_{3,1},\bx_2)$. If $\by_{3,1}$ is good, the $CF$ sends the bit $b=0$ to transmitter node~1 at time step $n+1$; otherwise the $CF$ sends the bit $b=1$. 
Now node~1, knowing bit $b$, sends codeword $\bx_1$ over the $n$ time steps of phase 2 if $b=0$ and $\bar{\bx}_1$ if $b=1$.

Let $\bx_1(w)$ and $\bar{\bx}_1(w)$ be the codewords corresponding to message $w=M_1^2$ of transmitter node~1 and let $\bx_2(v)$ be the codeword corresponding to message $v=M_2^2$ of transmitter node~2. If the bit sent from the $CF$ to transmitter node~1 in time step $n+1$ is 0, then the received symbols at node~3 are $(\by_{3,1},\bx_2(v))=W^n({\bx}_1(w),\bx_2(v))$, and if 
the bit sent from the $CF$ to transmitter node~1 in time step $n+1$ is 1, then the received symbols at node~3 are $(\by_{3,1},\bx_2(v))=W^n(\bar{\bx}_1(w),\bx_2(v))$.

\begin{claim}
For any $\bx_1(w)$ and $\bx_2(v)$, if $W_1^n(\bx_1(w),\bx_2(v))$ is not good then $W_1^n(\bar{\bx}_1(w),\bx_2(v))$ is good. 
\end{claim}

\begin{proof}
Follows from the definition of $W^n$.
Consider an entry $x_1$ in $\bx_1(w)$ and its corresponding entry $\bar{x}_1$ in $\bar{\bx}_1(w)$.
Let $x_2 \in \{0,1\}$.
It holds (by a simple exhaustive case analysis) that exactly one of the values $W_1(x_1,x_2)$ and $W_1(\bar{x}_1,x_2)$ is in the set $\{c,C\}$. Thus, the number of entries in $W_1^n(\bx_1(w),\bx_2(v))$ which are in the set $\{c,C\}$ plus the number of entries in $W_1^n(\bar{\bx}_1(w),\bx_2(v))$ which are in the set $\{c,C\}$ is exactly $n$. If the former is more than $n/2$ (implied by the fact that $W_1^n(\bx_1(w),\bx_2(v))$ is not good) then the latter is less than $n/2$ (implying that $W_1^n(\bar{\bx}_1(w),\bx_2(v))$ is good). 
\end{proof}

We now show that with high probability over our code design, for any good received word $(\by_{3,1},\by_{3,2})$, there are at most $\lceil 8/\delta \rceil$ codeword pairs $(\bx_1,\bx_2)$  that satisfy $W^n(\bx_1,\bx_2)=(\by_{3,1},\by_{3,2})$. Here $(\by_{3,1},\by_{3,2})$ is good iff $\by_{3,1}$ is good.
As $\by_{3,2}=\bx_2$, we analyze the number of possible codewords $\bx_1$ that satisfy $W^n(\bx_1,\bx_2)=(\by_{3,1},\bx_2)$.

\begin{claim}
\label{claim:good}
For any sufficiently large $n$, with high probability over our code design, for any good $\by_{3,1}$, there are at most $\lceil 8/\delta \rceil$ codewords $\bx_1 \in \{\bx_1(w)\}_w \cup \{\bar{\bx}_1(w)\}_w$ such that $W^n(\bx_1,\bx_2)=(\by_{3,1},\bx_2)$. 
\end{claim}

\begin{proof}
Let $\by_{3,1}$ be good.
Recall that the codebook $\{\bx_1(w)\}_w$ is chosen uniformly from subsets of size $|\cM_1^2|$ in $\cX_1^n$.
Consider choosing the codebook $\{\bx_1(w)\}_w$ in an iterative manner, where in iteration $w$, $\bx_1(w)$ is chosen uniformly from $\cX_1^n \setminus \{\bx_1(w')\}_{w' <w}$, i.e.,  the set $\{\bx_1(w)\}_w$ is chosen uniformly without repetitions. 
For any $w$, the probability (over the choice of $\bx_1(w)$) that $W_1^n(\bx_1(w),\bx_2)=\by_{3,1}$ is exactly the probability that $\bx_1(w) \in \bX_1(\by_{3,1})$ which is at most $\frac{2^{n/2}}{4^n-2^{3n/2-\delta n}} \leq 2\cdot 2^{-3n/2}$ for sufficiently large $n$.
For any $w$, the probability that $W_1^n(\bar{\bx}_1(w),\bx_2)=\by_{3,1}$ is exactly the probability that $\bar{\bx}_1(w) \in \bX_1(\by_{3,1})$ which is again at most $2 \cdot 2^{-3n/2}$.
Thus the probability (over the choice of $\bx_1(w)$) that $W_1^n(\bx_1(w),\bx_2)=\by_{3,1}$ or $W_1^n(\bar{\bx}_1(w),\bx_2)=\by_{3,1}$ is at most $4 \cdot 2^{-3n/2}$.
Moreover, the probability that there exist $\ell$ messages $\{w_1,\dots, w_\ell\}$ in $\cM_1^2$ such that for all $i=1,\dots,\ell$ it holds that either $W_1^n(\bx_1(w_i),\bx_2)=\by_{3,1}$ or $W_X^n(\bar{\bx}_1(w_i),\bx_2)=\by_{3,1}$ is at most 
\begin{eqnarray*}
 {{2^{3n/2 - \delta n}} \choose {\ell}} 4^\ell 2^{-3\ell n/2} & \leq &  2^{3n\ell/2 - \delta n\ell} 4^\ell 2^{-3\ell n/2} \\
& = & 4^\ell 2^{- \delta n\ell}
\end{eqnarray*}
Setting $\ell=\lceil 8/\delta \rceil+1$, we have that the above probability is at most $2^{-7 n}$ for sufficiently large $n$.
Union bounding over all $2^{6n}$ possible good $\by_{3,1}$, we conclude the assertion with probability $1-2^{-n}$ over code design.
\end{proof}

Assume a codebook $\{\bx_1(w)\}_w$ that satisfied the conditions of Claim~\ref{claim:good}. We now complete the description of our encoding and decoding scheme of phase 2. 
Let $\bx_1(w)$, $\bx_2$ be the codewords forwarded to the $CF$ during the first phase of communication and let $b$ be the one bit value sent back to transmitter node~1 (from the $CF$) at time step $n+1$.
If $b=0$, let $\by_{3,1}=W_1^n(\bx_1(w),\bx_2)$ otherwise let $\by_{3,1}=W_1^n(\bar{\bx}_1(w),\bx_2)$.
The vector $\by_{3,1}$ is received at receiver node~3 during phase 2.
We are guaranteed that $\by_{3,1}$ is good. 
By Claim~\ref{claim:good} we know that at most $\lceil 8/\delta \rceil$ other codewords $\bx_1$ may also satisfy $\by_{3,1}=W_1^n(\bx_1,\bx_2)$. Thus receiver node~3 can decode to a list $L({\by_{3,1}})$ of codewords (i.e., messages) corresponding to $\by_{3,1}$ of size at most $\lceil 8/\delta \rceil$ that includes the codeword $\bx_1(w)$ or $\bar{\bx}_1(w)$ sent by transmitter node~1.
 
To decode the message $M_1^2$ correctly,  receiver node~3 needs to know whether $\bx_1(M_1^2)$ or $\bar{\bx}_1(M_1^2)$ was sent by node~1 and in addition the index of $\bx_1(M_1^2)$ or $\bar{\bx}_1(M_1^2)$ in the list $L(\by_{3,1})$. Here we use the fact that the codebook $\{\bx_1(w)\}_w$ consists of distinct codewords.
This information can be sent from the $CF$ to receiver node~3 on the $k$, 1-bit channels of $\cN^{+}$.
Here we have that $k=\lceil \log{(8/\delta)} \rceil+1$.
Specifically, at time step $n+1$, the $CF$ can simulate the behavior of the channel $W^n$ (as $W$ is deterministic) and thus deduce $\by_{3,1}$ and the list $L({\by_{3,1}})$ available to receiver node~3. Using a lexicographic ordering on all vectors of $\cY_{3,1}^n$, at time step $n+1$, the $CF$ can send the bit $b$ (the same bit it sends to transmitter node~1) and also the location of  $\bx_1(M_1^2)$ or $\bar{\bx}_1(M_1^2)$ in $L(\by_{3,1})$ to receiver node~3. This information allows node~3 to determine $M_1^2$ (without error).
This concludes the proof for phase 2 of our communication.

\subsection{Putting it all together}
All in all, in $2n+1$ time steps, the message $M_1=(M_1^1,M_1^2) \in [2^n]\times [2^{\frac{3n}{2}-\delta n}]$, and the message $M_2 = (M_2^1,M_2^2) \in [2^{n}]\times[2^{n}]$ are decoded with zero error at receiver node~3. This implies a communication rate of 
$$
\left(\frac{2n(1.25-0.5\delta)}{2n+1},\frac{2n}{2n+1}\right)
$$ 
which in turn implies for $\delta=0.02$ that
$$
\left(1.24,1\right) \in \Cmax(\cN^+)
$$
as asserted in the theorem.

\section{Conclusions}
In this work we study the effect of 1-bit of communication in memoryless networks.
We prove the existence of a network $\cN$ and a 1-bit channel $e$ such that adding $e$ to $\cN$ yields a new network $\cN^*$ with a strictly larger capacity region. 
Our results hold for the maximal-error criteria.
Whether there exist a memoryless network $\cN$ and a $\delta$ capacity noiseless channel $e$ (let alone a 1-bit channel $e$) such that adding $e$ to $\cN$ results in a capacity region $\cC_{avg}(\cN(\delta))$ which is not continuous at $\delta=0$ is an intriguing open problem. 
The effect of adding a 1-bit channel in the noiseless setting of network coding is not addressed in this work and remains a fascinating open problem.

\bibliographystyle{unsrt}
\bibliography{proposal}

\begin{thebibliography}{10}

\bibitem{HEJ:10}
T.~Ho, M.~Effros, and S.~Jalali.
\newblock On equivalences between network topologies.
\newblock In {\em Forty-Eighth Annual Allerton Conference on Communication,
  Control, and Computing}, 2010.

\bibitem{JEH:11}
S.~Jalali, M.~Effros, and T.~Ho.
\newblock On the impact of a single edge on the network coding capacity.
\newblock In {\em Information Theory and Applications Workshop (ITA)}, 2011.

\bibitem{guruswami2003list}
V.~Guruswami.
\newblock List decoding with side information.
\newblock In {\em Proceedings of 18th IEEE Annual Conference on Computational
  Complexity}, pages 300--309, 2003.

\bibitem{langberg2004private}
M.~Langberg.
\newblock Private codes or succinct random codes that are (almost) perfect.
\newblock In {\em Annual IEEE Symposium on Foundations of Computer Science},
  volume~45, pages 325--334, 2004.

\bibitem{NEL:16}
P.~Noorzad, M.~Effros, and M.~Langberg.
\newblock Can negligible cooperation increase network reliability?
\newblock In {\em International Symposium on Information Theory}, 2016.

\bibitem{Dueck:78}
G.~Dueck.
\newblock Maximal error capacity regions are smaller than average error
  capacity regions for multi-user channels.
\newblock {\em Probl. Contr. Inform. Theory}, 7:11--19, 1978.

\bibitem{CG10}
T.~Chan and A.~Grant.
\newblock On capacity regions of non-multicast networks.
\newblock In {\em International Symposium on Information Theory}, pages 2378 --
  2382, 2010.

\bibitem{chan2014network}
T.~Chan and A.~Grant.
\newblock Network coding capacity regions via entropy functions.
\newblock {\em IEEE Transactions on Information Theory}, 60(9):5347--5374,
  2014.

\bibitem{LE11}
M.~Langberg and M.~Effros.
\newblock Network coding: Is zero error always possible?
\newblock In {\em Proceedings of Forty-Ninth Annual Allerton Conference on
  Communication, Control, and Computing}, pages 1478--1485, 2011.

\bibitem{LE:12}
M.~Langberg and M.~Effros.
\newblock {Source coding for dependent sources.}
\newblock {\em In proceedings of IEEE Information Theory Workshop (ITW)}, 2012.

\bibitem{WLE13}
M.~F. Wong, M.~Langberg, and M.~Effros.
\newblock {On a Capacity Equivalence between Network and Index Coding and the
  Edge Removal Problem}.
\newblock In {\em proceedings of International Symposium on Information
  Theory}, pages 972 -- 976, 2013.

\bibitem{WLE:15}
M.~F. Wong, M.~Langberg, and M.~Effros.
\newblock On an equivalence of the reduction of k-unicast to 2-unicast capacity
  and the edge removal property.
\newblock In {\em IEEE International Symposium on Information Theory (ISIT)},
  pages 371--375, 2015.

\bibitem{WLE:16}
M.~F. Wong, M.~Langberg, and M.~Effros.
\newblock On the tightness of an entropic region outer bound for network coding
  and the edge removal property.
\newblock {\em To appear in IEEE International Symposium on Information
  Theory}, 2016.

\end{thebibliography}

\end{document}